\documentclass[11pt]{article}
\usepackage{verbatim,url,enumerate,color,paralist}
\usepackage{amsmath,amsfonts,amssymb,amstext,amsthm}
\usepackage{fullpage}
\usepackage{hyperref}
\usepackage[margin=2cm]{geometry}

\newtheorem{theorem}{Theorem}
\newtheorem{lemma}[theorem]{Lemma}
\newtheorem{corollary}[theorem]{Corollary}

\newtheorem{claim}[theorem]{Claim}

\newtheorem{remark}[theorem]{Remark}

\newcommand{\R}{\ensuremath{\mathbb{R}}}
\newcommand{\supp}{\mbox{supp}}
\begin{document}

\title{A Rational Convex Program for Linear Arrow-Debreu Markets}
\author{Nikhil R. Devanur\thanks{  {Microsoft Research, Redmond, USA.}
     { (nikdev@microsoft.com)}}
 \and  Jugal Garg\thanks{  {College of Computing, Georgia Institute of
   Technology, Atlanta, USA.}
     {  (jgarg@cc.gatech.edu)}}
  \and L\'aszl\'o A. V\'egh\thanks{
       {Department of Management,}
       {London School of Economics, London, UK.}
      {(l.vegh@lse.ac.uk)}}
}
\date{}
\maketitle
\begin{abstract}
We give a new, flow-type convex program describing equilibrium solutions to linear Arrow-Debreu markets. 
Whereas convex formulations were previously known (\cite{nenakov83,jain,cornet89}), our program exhibits several new features.
It gives a simple necessary and sufficient condition and a concise
proof of the existence and rationality of equilibria, settling an open question raised by Vazirani \cite{vazirani12}.
As  a consequence we also obtain a simple new proof of Mertens's \cite{mertens03} result that the equilibrium prices form a convex polyhedral set.
\end{abstract}

\section{Introduction}

The {\em exchange market model} is a classical model of a market along with a notion of equilibrium,
 introduced by Walras in 1874 \cite{walras}. In this model, agents arrive at the market with an initial endowment of
divisible goods, and a  utility function for consuming goods.
A market equilibrium assigns prices to the goods such that when every agent uses the revenue from selling
her initial endowment for  purchasing a bundle of goods that maximizes
her utility, the market {\em clears}, i.e, the total demand for every good is equal to its supply. 
The celebrated theorem by Arrow and Debreu \cite{arrow}
proves the existence of a market equilibrium under mild necessary conditions on the 
utility functions - therefore it is commonly known as the Arrow-Debreu
market model. Since then, understanding equilibrium behavior and
computing equilibrium prices has been extensively studied 
in mathematical economics and more recently in theoretical computer
science.

In this paper we study the {\em linear Arrow-Debreu model}, where the utility functions of agents are {\em linear}.
Let us first mention results
pertaining to a well-studied further special case,
the {\em linear
  Fisher model}, that was formulated by Fisher in 1891, who also studied
 the computability of equilibrium,  via a hydraulic
machine no less! (See Brainard and Scarf \cite{brainard} for a fascinating account.) In this model, the
agents are separated into two types, buyers and sellers; buyers arrive
to the market with a certain amount of money they wish to spend on
goods offered by the sellers. 
This model turned out to be
substantially easier from a computational perspective  than the linear Arrow-Debreu model.
A convex programming formulation was given by Eisenberg and Gale in
1959 \cite{eisenberg59}. The problem of equilibrium computation was introduced
to the theoretical computer science community by Devanur et
al. \cite{devanur08}, who gave a polynomial time combinatorial
primal-dual algorithm. This initiated an intensive line of research, most notable among which is a 
strongly polynomial time algorithm by Orlin \cite{orlin}; for a survey, see
\cite[Chapter 5]{nisan07} or \cite{vazirani12}.
Also, Shmyrev \cite{Shmyrev09} gave a new type of convex
program (which was discovered independently by  Birnbaum et al. \cite{BDX}) capturing the equilibria. 

Let us now turn to the linear Arrow-Debreu model. The first important
algorithmic result was a finite Lemke-type path following algorithm for finding
an equilibrium solution by Eaves \cite{eaves}.  A remarkable
consequence of this algorithm is that  when the utilities are given by rational numbers, 
there also exists an equilibrium among rational numbers. 

The history of convex programming formulations for the linear Arrow-Debreu model is somewhat convoluted.
Jain \cite{jain}  formulated a convex
program whose solutions correspond to market equilibria;
this can be used to obtain a polynomial time algorithm via the
Ellipsoid algorithm. It turned out later that the same convex program
was already formulated by Nenakov and Primak in 1983 \cite{nenakov83}.
Interestingly, the computer science community so far seems to have been unaware of the 1989 paper 
 by Cornet in \cite{cornet89} giving a similar, yet better convex program.
This is not mentioned even in the survey paper by Codenotti et al. \cite{codenotti04} exploring the background of the problem.
 (These convex programs will be exhibited in Section~\ref{sec:general}.)

An unsatisfactory aspect of the program in \cite{jain,nenakov83}
is that it fails to show the existence of an equilibrium; 
it only shows that {\em if} there exists an equilibrium, then any feasible solution to the convex program is one. In contrast, 
Cornet's program provides a proof of existence assuming that a stronger sufficient condition given by Gale in 1957
\cite{gale57} holds, however fails to show it for the weaker necessary and sufficient condition given by Gale in 1976
\cite{gale}. 
An efficient interior point algorithm to compute an equilibrium was given by Ye \cite{Ye08} based on the convex program in \cite{jain,nenakov83}.
An important recent result is a combinatorial primal-dual algorithm by Duan and Mehlhorn \cite{duan13}; this does not rely on the convex programming formulation but adapts techniques from the algorithm by Devanur et al. \cite{devanur08} for linear Fisher markets.

The convex programs \cite{nenakov83,jain} and \cite{cornet89} for the linear Arrow-Debreu model
is of substantially different nature from those \cite{eisenberg59,Shmyrev09,BDX} for the linear Fisher
model. The latter ones have linear constraints only, with separable
convex objectives, in contrast to the nonlinear constraints in
\cite{nenakov83,jain}. Whereas \cite{cornet89} is formulated with only very simple linear constraints,
the max-min type objective in fact hides similar nonlinear constraints.

The feasible region for both formulations for the linear Fisher model are indeed classical polyhedra,
\cite{eisenberg59} a generalized flow polyhedron and \cite{Shmyrev09,BDX}
a circulation polyhedron. This also explains why for the linear Fisher model, classical flow
techniques are applicable (see  \cite{vegh12a,vegh12b}) and strongly polynomial time algorithms exist. 
Also interestingly, the convex
programs of \cite{eisenberg59} and \cite{Shmyrev09,BDX} fall into the class
of {\em rational convex programs}, defined by Vazirani
\cite{vazirani12}: for a rational input, there exists a rational
optimal solution with bitsize bounded polynomially in the input size.
For the previous programs \cite{cornet89,nenakov83,jain} the proof of
the existence of a rational optimal solution requires further nontrivial arguments (e.g. \cite{eaves,duan13}).
 An open problem in
\cite{vazirani12} asks for the existence of a rational convex program
for the linear Arrow-Debreu model with a simple, direct proof of rationality.

In this paper, {\em we exhibit a rational convex program for the linear Arrow-Debreu model,
that also guarantees the existence of an equilibrium}, 
thus settling the open questions of \cite{vazirani12}. 
Our convex program draws from the convex programs in \cite{eisenberg59,Shmyrev09,BDX};  
more precisely, it is a combination of the convex program of \cite{Shmyrev09,BDX} and the dual of the convex program of \cite{eisenberg59} (see Devanur \cite{devanur-convex}).
The objective function has terms from both the convex programs and 
there are two sets of constraints, one describing a circulation polyhedron as in \cite{Shmyrev09,BDX}, 
and another that is similar to the dual of the Eisenberg-Gale \cite{eisenberg59} convex program. 
The main technical contribution is to show the existence of an
equilibrium based on the Karush-Kuhn-Tucker (KKT) conditions for this
convex program.
Our program is feasible if and only if Gale's \cite{gale} necessary and sufficient
conditions on the existence of  equilibria hold. The
existence of a rational optimal solution for rational input follows by showing that there exists an optimal solution that is an extreme point of the feasible region.

\medskip

Now we give a formal description of the model and give our convex program.
We are given set $A$ of $n$ agents, and
 assume that there is a one-to-one mapping between agents and goods,
every agent $i\in A$ arrives with one divisible unit of good of type $i$.
This is without loss of generality:
the general case with an arbitrary set of goods and arbitrary initial
endowments can be easily
 reduced to this setting; see Section~\ref{sec:general}.
  The utility of agent $i$ for the
 good of agent $j$ is $u_{ij}\ge 0$. 
The directed graph $(A,E)$ contains an arc $ij$ for every pair with
$u_{ij}>0$; it may also contain loops expressing that some agents are
interested in their own goods. 
We make the standard assumption that for each agent $i\in A$, $E$ contains at least one incoming and one outgoing arc incident to $i$.
 
By a {\em market equilibrium}, we mean a set of prices $p:A\rightarrow \R_+$
and allocations $x:E\rightarrow \R_+$ satisfying the following conditions.
\begin{itemize}
\item {\em Market clearing:} Demand equals supply.
\begin{itemize}
\item $\sum_{i\in A} x_{ij}=1$, for every $j\in A$, i.e., every good is fully sold.
\item $p_i=\sum_{j\in A}x_{ij}p_j$  for every $i\in A$, i.e., the money spent by
  agent $i$ equals his income $p_i$.
\end{itemize}
\item {\em Optimal bundle:} Every agent is allocated a utility maximizing bundle
subject to its budget constraint. That reduces to
\begin{itemize}
\item For every $i\in A$, if  $x_{ij}>0$ then $u_{ij}/p_j$ is the maximal
value over $j\in A$.
\item $p_i>0$ for every $i\in A$;
\end{itemize} 
\end{itemize}

It is easy to see that the following condition is necessary for the existence of an equilibrium:
\begin{equation}\tag{$\star$}\label{cond:suff}
\begin{aligned}
\mbox{For every strongly connected component $S\subseteq E$ of the digraph $(A,E)$,}\notag\\
\mbox{if $|S|=1$ then there is a loop incident to the node in $S$.}
\end{aligned}
\end{equation}
Indeed, assume $\{k\}$ is a singleton strongly connected component
without a loop. Let $T$ denote the set of nodes different from $k$ that can be reached on
a directed path in $E$ from $k$. In an equilibrium
allocation, the agents in $T\cup \{k\}$ spend all their money on the
goods of the agents in $T$; this implies $p_k=0$, contrary to our assumption.

We formulate the following convex program, with variables $p_i$ representing the prices, the $\beta_i$'s the inverse best bang-per-bucks,
and $y_{ij}$ the money paid by agent $i$ to agent $j$.
\begin{equation}\tag{CP}\label{CP}
\begin{aligned}
\min \sum_{i\in A} p_i\log\frac{p_i}{\beta_i}&-\sum_{ij\in E}y_{ij}\log u_{ij}\notag\\
\sum_{i:ij\in E}y_{ij}&=p_j \quad \forall j\in A\notag\\
\sum_{j:ij\in E}y_{ij}&=p_i \quad \forall i\in A\\
u_{ij}\beta_i&\le p_j\quad \forall ij\in E\notag\\
p_i&\ge 1\quad\forall i\in A\notag\\
y,\beta&\ge 0\notag
\end{aligned}
\end{equation}

\begin{theorem}\label{thm:main}
Consider an instance of the linear Arrow-Debreu market given by the graph
$(A,E)$ and the utilities $u:E\rightarrow\R_+$. 
The convex program (\ref{CP}) is feasible if and only if
(\ref{cond:suff}) holds, and in this case the optimum value is 0,
and the prices $p_i$ in an optimal solution give a market equilibrium with allocations $x_{ij}=y_{ij}/p_j$.
Further,  if all utilities are rational numbers, then there exists a market
 equilibrium with all prices and allocations also rational, of bitsize
 polynomially bounded in the input size.
\end{theorem}
Here, the bitsize of the rational number $p/q$ is defined as $\lceil
\log_2 p\rceil+\lceil \log_2 q\rceil$. 
The rational optimum property follows by observing that there exists an optimal extremal point solution. 
The following results easily follow from the above theorem:
\begin{corollary}
The following hold for linear Arrow-Debreu markets.
\begin{enumerate}[(i)]
\item For every agent the utility is the same at every equilibrium.
\item The vectors $(y,p)$ at equilibrium form a convex set. In particular, the set of price vectors at equilibrium is convex.
\end{enumerate}
\end{corollary}
Property {\em (i)} was already proved by Gale \cite{gale} in 1976 and also follows from Cornet \cite{cornet89}.
Whereas the convexity of equilibrium prices was proved by Mertens \cite{mertens03} and by Florig \cite{florig04}, both these proofs are quite involved, whereas it is a straightforward consequence of Theorem~\ref{thm:main}.  We are not aware of previous proofs on the convexity of $y$. In contrast,
Cornet \cite{cornet89} proved that $(x,\log p)$ is convex at equilibria; here $x_{ij}=y_{ij}/p_j$ is the amount of good $j$ allocated to 
agent $i$.

The Lagrangian dual of (\ref{CP}) is similar to
Cornet's program \cite{cornet89} (see (\ref{CPC}) in
Section~\ref{sec:general}) but is different from it. Also, analyzing the optimal
Lagrangian multipliers for (\ref{CP}) we can derive the
feasibility convex program  \cite{nenakov83,jain}; these
correspondences will be explained in Section~\ref{sec:general}. Our program exhibits some new and advantageous features as compared to Cornet's:
\begin{itemize}
\item The program (\ref{CP}) provides necessary and sufficient condition of the existence of equilibria. In contrast,
Cornet's program provides only a stronger sufficient condition given by Gale \cite{gale57}.  
\item The program (\ref{CP}) is feasible if and only if there exists an equilibrium. In contrast, Cornet's program can be feasible also if there exists no equilibrium; in this case the objective is unbounded.
\item The program (\ref{CP}) also demonstrates the existence of a rational equilibrium, which is not the case with Cornet's program.
\item All constraints in (\ref{CP}) are linear.
\item Our program establishes links to known convex programs for the Fisher model.
\end{itemize}
We think that the discovery of this convex program will pave the way for more efficient (and in particular, 
strongly polynomial time) algorithms for this model. 

\medskip

The rest of the paper is structured as follows.
 Section~\ref{sec:main-proof} is
dedicated to the proof of Theorem~\ref{thm:main}. This is based on the
KKT conditions, however, the argument is not
straightforward, in contrast to similar arguments for the convex
programs of \cite{eisenberg59,Shmyrev09,BDX}. 
 Section~\ref{sec:general}
shows the equivalence of our existence condition (\ref{cond:suff})
to previous results by Gale \cite{gale57,gale}, 
 exhibits the previous convex programs \cite{cornet89,nenakov83,jain}, and
explains the correspondence between our formulation (\ref{CP}) and these programs.


\medskip

\section{Proof of Theorem~\ref{thm:main}}\label{sec:main-proof}
Let us first verify that (\ref{CP}) is actually a convex program. The
feasible region is defined by linear constraints, so we only have to
check that the objective is convex. The terms corresponding to the
$y_{ij}$'s are linear. The term $\sum_{i\in
  A}p_i\log\frac{p_i}{\beta_i}$ is the relative entropy of $p$ and
$\beta$ and is well-known to be convex in the nonnegative variables $p_i$ and $\beta_i$.\footnote{%
Let us give a simple proof. We need to verify that for every
$q,b,q',b'\ge 0$ and $0<\lambda<1$, we have
\[
\lambda q\log \frac{q}b+(1-\lambda) q'\log \frac{q'}{b'}\ge (\lambda q+(1-\lambda)q')\log \frac{\lambda q+(1-\lambda)q'}{\lambda b+(1-\lambda)b'}.
\]
This can be derived using the convexity of $x\log x$ for $q/b$, $q'/b'$ with the linear combination $\lambda^*=\frac{\lambda b}{\lambda b+(1-\lambda)b'}$.}
Let us now verify the feasibility claim.
\begin{claim}\label{cl:suff}
The convex program (\ref{CP}) is feasible if and only if (\ref{cond:suff}) holds.
\end{claim}
\begin{proof}
Assume that (\ref{cond:suff}) is violated, that is, there is a strongly connected component consisting of a single node $i_0$, and there is no loop in $E$ incident to $i_0$ (that is, $u_{i_0i_0}=0$.)
For a contradiction, assume (\ref{CP}) admits a feasible solution
$(y,p,\beta)$. Then $y$ gives a feasible circulation on the graph
$(A,E)$ such that there is a positive amount of flow entering (and
leaving) every node.
%
The circulation $y$ can be decomposed to a weighted sum of directed
cycles: $y=\sum_{k=1}^t w_k\chi_{C_k}$, where for each $1\le k\le t$,
$\chi_{C_k}$ is the 0-1 incidence vector of a directed cycle $C_k$,
and $w_k\ge 0$.
Clearly every cycle  $C_k$ must be contained inside a strongly
connected component. Hence no cycle may be incident to $i_0$, that is,
the flow entering this node is 0, a contradiction.

Assume now that (\ref{cond:suff}) is satisfied. Consequently, there is
a directed cycle $C_i$ in $(V,A)$ incident to every node $i$.
Set $y=\sum_{i\in A}\chi_{C_i}$, and let $p_i$ denote the amount of
$y$ entering the node $i$.
This gives a feasible solution to (\ref{CP}) with $\beta_i=\min_{j\in A}\frac{p_j}{u_{ij}}$.
\end{proof}

\begin{claim}\label{claim:non-neg}
The objective value in \ref{CP}) is non-negative, and it is 0 if and
only if the prices $p_i$ and the allocations  $x_{ij}=y_{ij}/p_j$ form
a market equilibrium. Conversely, for every market equilibrium
$p'_i,x'_{ij}$, we get an optimal solution to \ref{CP}) by setting
$p_i=\alpha p'_i$, $y_{ij}=\alpha p_jx_{ij}$, $\beta_i=\min_{j\in A} \alpha p'_j/u_{ij}$, where 
$\alpha=1/\min\{1,\min_{i\in A} p_i\}$.
\end{claim}
\begin{proof}
By the third inequality, $-\log u_{ij}\ge \log \beta_i-\log p_j$.
Hence the second term in the objective is at least 
\begin{align*}
\sum_{ij\in E} (\log \beta_i-\log p_j)y_{ij}=\sum_{i\in A} \log \beta_i \left(\sum_{j:ij\in E} y_{ij}\right) -\sum_{j\in A} \log p_j \left(\sum_{j:ij\in E} y_{ij}\right)=\\
\sum_{i\in A}p_i\log \beta_i-\sum_{j\in A}p_j\log p_j=-\sum_{i\in A}p_i\log \frac{p_i}{\beta_i}.
\end{align*}
This implies that the objective value is $\ge 0$. Moreover, the lower bound is tight if and only if $u_{ij}\beta_i=p_j$ whenever $y_{ij}>0$. This is equivalent to all transactions being best bang-per-buck purchases. It is easy to verify that the solution represents a market equilibrium.
The second part also follows easily.
\end{proof}

The proof of the  assertion  in Theorem~\ref{thm:main} that optimal
solutions to (\ref{CP}) correspond to market equilibria is complete by the following lemma.
\begin{lemma}\label{lem:main-0}
Whenever (\ref{CP}) is feasible, the optimum value is 0.
\end{lemma}
Let us now formulate the Karush-Kuhn-Tucker conditions on
optimality. Since all constraints in (\ref{CP}) are linear, these are
sufficient and feasible for optimality.
Consider an optimal solution $(p,y,\beta)$, and 
let us associate Lagrangian multipliers $\delta_j$, $\gamma_i$, $w_{ij}$ and $\tau_i$ to the inequalities in the order as described in (\ref{CP}).
We obtain the following conditions.
\begin{align}
-\delta_j+\gamma_i&\le -\log u_{ij}\quad\forall ij\in E\label{kkt-1}\\
\delta_i-\gamma_i+\sum_{j:ji\in E}w_{ji}+\tau_i&=\log \frac{p_i}{\beta_i} +1\quad\forall i\in A\label{kkt-2}\\
-\sum_{j:ij\in E} u_{ij}w_{ij}&\le -\frac {p_i}{\beta_i}\quad \forall i\in A\label{kkt-3}
\end{align}
Also, (\ref{kkt-1}) must be tight for all $y_{ij}>0$, and (\ref{kkt-3})
must be tight for all $\beta_i>0$. Further, $\tau_i>0$ implies
$p_i=1$, and 
$w_{ij}>0$ implies $u_{ij}\beta_i=p_j$. Note that in an
optimal solution every $\beta_i>0$, and hence (\ref{kkt-3}) always
holds with equality.
We can therefore derive the following from (\ref{kkt-3}):
\begin{equation}
p_i=\sum_{j:ij\in E}u_{ij}\beta_i w_{ij}=\sum_{j:ij\in E}p_j w_{ij}.\label{kkt-3b}
\end{equation}

The following remark can be interpreted as a ``self-duality''
property: a market equilibrium does not only provide a primal optimal
solution to (\ref{CP}) but also optimal Lagrangian multipliers.
\begin{remark}\label{rem:self-dual}
Assume there exists a market equilibrium $(p,x)$; by re-scaling, we may
assume $p_i\ge 1$ for all $i\in A$. As in Claim~\ref{claim:non-neg},
$p$, $y_{ij}=p_jx_{ij}$ and $\beta_i=\min_{i\in A} p_j/u_{ij}$ give an optimal solution to (\ref{CP}). 
It is straightforward to check that $\gamma_j=\log \beta_j$,
$\delta_j=\log p_j$, $w_{ij}=x_{ij}$  and $\tau\equiv 0$ give optimal
Lagrangian multipliers.
\end{remark}

The next claim expresses the optimum objective value of (\ref{CP}) in
terms of the Lagrangian multipliers.
\begin{claim}
Let $(y,p,\beta)$ be a primal optimal solution, and let $(\gamma,\delta,w,\tau)$ be optimal Lagrangian multipliers. Then
\[
\sum_{i\in A} p_i\log\frac{p_i}{\beta_i}-\sum_{ij\in E}y_{ij}\log u_{ij}=\sum_{i\in A}\tau_i
\]
\end{claim}
\begin{proof}
By complementary slackness, (\ref{kkt-1}) is tight whenever $y_{ij}>0$. Taking the combination of these equalities multiplied by $y_{ij}$, we get
\[
-\sum_{ij\in E}y_{ij}\log u_{ij}=\sum_{ij\in E}y_{ij}(\gamma_i-\delta_j)=\sum_{i\in A}(\gamma_{i}-\delta_i)p_i.
\]
In the second equality, we used the degree constraints in (\ref{CP}).
Next, let us add the equalities (\ref{kkt-2}) multiplied by $p_i$. We obtain
\begin{align*}
\sum_{i\in A}\left(p_i\log\frac{p_i}{\beta_i}+p_i\right)=
\sum_{i\in A}(\delta_i-\gamma_i)p_i+\sum_{i\in A}\sum_{j:ji\in E}w_{ji}p_i+\sum_{i\in A}\tau_ip_i=\\
\sum_{i\in A}(\delta_i-\gamma_i)p_i+\sum_{i\in A}p_i+\sum_{i\in A}\tau_i.
\end{align*}
Here we used (\ref{kkt-3b}) for the second term, and that $p_i=1$ whenever $\tau_i>0$ for the third term. Adding this to the previous inequality proves the claim.
\end{proof}

Using the previous claim,  Lemma~\ref{lem:main-0} follows from the next lemma.
\begin{lemma}\label{lem:tau-0}
For the optimal Lagrange multipliers $(\gamma,\delta,w,\tau)$, it follows that
\[
\tau_i=0\quad \forall i\in A.
\]
\end{lemma}
\begin{proof}
The proof is by induction on the number of agents $|A|$. We assume
that for all markets with $<|A|$ agents, the assertion holds.
Let us introduce $q_i:=e^{\delta_i}$ and
$\theta_i:=e^{\gamma_i}$. These are quantities playing a similar role to $p_i$ and $\beta_i$: the conditions (\ref{kkt-1}) can be rewritten as
\[
u_{ij}\theta_i\le q_j\quad \forall ij\in E,
\]
and furthermore by complementary slackness it follows that if
$y_{ij}>0$ then equality must hold. The $\theta_i$'s are therefore the inverse best bang-per-buck values for the prices $q$. Let $F\subseteq E$ denote the set of arcs with
$u_{ij}\beta_i=p_j$ and $H\subseteq E$ the set of arcs with
$u_{ij}\theta_i=q_j$. By complementary slackness, $\supp(y)\subseteq H$ and $\supp(w)\subseteq F$. 
Let us define
\[
\alpha:=\max_{i\in A}\frac{q_i}{p_i},\quad S:=\left\{i\in A: \frac{q_i}{p_i}=\alpha\right\}.
\]
\begin{claim}\label{cl:alpha}
We have $\frac{p_i}{\beta_i} \le \frac{q_i}{\theta_i}$ for every $i\in
S$. Further, if $ij\in F$, $i\in S$   and
$\frac{p_i}{\beta_i}=\frac{q_i}{\theta_i}$, then $j\in S$ holds.
\end{claim}
\begin{proof}
The first claim is equivalent to $\frac{\theta_i}{\beta_i}\le \alpha$ if $i\in S$. 
This follows since
\[
\theta_i=\min_{j\in A}\frac{q_{j}}{u_{ij}}\le \min_{j\in A}\frac{\alpha p_{j}}{u_{ij}}=\alpha\beta_i.
\]
For the second part, assume for a contradiction that $q_j<\alpha p_j$
for some best bang-per-back arc $ij\in F$ with $i\in S$. This would imply that the inequality above is strict, giving a contradiction.
\end{proof}

Together with (\ref{kkt-2}), this gives 
\begin{eqnarray}\label{eq:tau}
\sum_{i\in A} w_{ij} \le 1 - \tau_j \quad \forall j\in S,
\end{eqnarray}
with equality only if $\frac{p_j}{\beta_j}=\frac{q_j}{\theta_j}$.
Let 
\[
T:=\{i\in A:  j\in S\ \forall ij\in F\}
\]
 denote the sets of agents having all their best bang-per-buck goods
 in $S$ with respect to prices $p$.
Recall that $\supp(w)\subseteq F$. By the definition of $T$, we get
from (\ref{kkt-3b}) that
\begin{eqnarray}\label{eq:wp}
\sum_{j\in S}w_{ij}p_j = p_i \quad \forall i \in T.
\end{eqnarray}
Combining this with the straightforward $\sum_{j\in S}y_{ij}\le p_i$,
for all $i\in T$, we obtain
\[
\sum_{i\in T}\sum_{j\in S} w_{ij}p_j \ge \sum_{i\in T}\sum_{j\in S} y_{ij}.
\]
Rearranging the sums gives
\begin{equation}\label{eq:rearrange}
\sum_{j\in S}p_j \sum_{i\in T} w_{ij} \ge \sum_{j\in S}\sum_{i\in T} {y_{ij}}
\end{equation}
The next step requires the following observation.
\begin{claim}\label{cl:A}
For every arc $ij\in H$ with $j\in S$, it follows that $i\in T$.
\end{claim}
\begin{proof}
For  a contradiction, assume $i\notin T$, that is, there exists a
good $j'\notin S$ with $ij'\in F$. Then
\[
 \theta_i=\frac{q_{j}}{u_{ij}}=\alpha\frac{p_{j}}{u_{ij}}\ge
 \alpha \beta_i=\alpha\frac{p_{j'}}{u_{ij'}}>\frac{q_{j'}}{u_{ij'}}\ge \theta_i,
 \]
 a contradiction.
\end{proof}

Recall that $\supp(y)\subseteq H$, and therefore if $j\in S$ and $y_{ij}>0$,
then $i\in T$ must hold by the above Claim. Hence if $j\in S$, then
$\sum_{i\in T} {y_{ij}}=p_j$. Combining this with
(\ref{eq:tau}) and  (\ref{eq:rearrange}), we get 
\begin{equation}
\hspace{-1cm}\sum_{j\in S} (1-\tau_j)p_j \ \ge \  \sum_{j\in S}p_j \sum_{i\in A} w_{ij} \ \ge \  \sum_{j\in S} p_j
\sum_{i\in T} w_{ij} \ \ge \  \sum_{j\in S}\sum_{i\in T} y_{ij} \  =
\sum_{j\in S}p_j.\label{eq:multi}
 \end{equation}
We must have equality throughout, and therefore for all $j\in S$ it
follows that $\tau_j=0$ and  $\frac{p_j}{\beta_j}=\frac{q_j}{\theta_j}$;
the latter was a necessary condition for equality in (\ref{eq:tau}).
Now the second part of 
  Claim~\ref{cl:alpha} guarantees that $S\subseteq T$.

Using (\ref{eq:wp}), we have $\sum_{i\in T}\sum_{j\in
  S}w_{ij}p_j=\sum_{i\in T}p_i$. On the other hand, the above
equalities guarantee $\sum_{i\in T}\sum_{j\in
  S}w_{ij}p_j=\sum_{i\in S}p_i$. We can
therefore conclude $S=T$. Moreover, the following holds.

\begin{claim}
No arc in $\supp(y)\cup \supp(w)$ enters or leaves the set $S$.
\end{claim}
\begin{proof}
Recall that $\supp(y)\subseteq H$ and $\supp(w)\subseteq F$.
Since $S=T$, the definition of $T$ implies that no arc $ij\in F$
leaves $S$; recall that $\supp(w)\subseteq F$. The second inequality in (\ref{eq:multi}) must hold with
equality, implying that $w_{ij}=0$, whenever
$i\in A\setminus S$, $j\in S$. Claim~\ref{cl:A} implies that no arc
$ij\in H$ enters $S$, and $\supp(y)\subseteq H$. The first to equalities in (\ref{CP}) imply that
$\sum_{i\in S, j\in A\setminus S} y_{ij}=\sum_{i\in A\setminus S, j\in
  S} y_{ij}$. Hence no arc with $y_{ij}>0$ may leave $S$. 
\end{proof}

If $A=S$, then the proof of Lemma~\ref{lem:tau-0} is complete. If
$S\subsetneq A$, then consider the 
restrictions of $(p,y,\beta)$ and $(\gamma,\delta,w,\tau)$ to $A\setminus
S$, and to the
arcs inside $A\setminus S$. The first gives a feasible solution to
(\ref{CP}) on the restricted graph, whereas the second give optimal
Lagrange multipliers, since the primal-dual slackness conditions are
satisfied. According to our assumption on $S$ being
a minimal counterexample, it follows that $\tau_i=0$ for all $i\in
A\setminus S$, completing the proof.
\end{proof}

To complete the proof of Theorem~\ref{thm:main}, it is left to verify
the claim on the existence of a rational optimal solution. This will follow from the next structural
observation; note that the feasible region is a polyhedron.

\begin{claim}\label{claim:extremal}
There exists an optimal solution to (\ref{CP}) that is an extremal
point of the feasible region.
\end{claim}
\begin{proof}
Consider an optimal solution $z=(p,y,\beta)$ to (\ref{CP}); by the
above, we know that it corresponds to a market equilibrium. As every
point in the feasible region, $z$ can be written as
the sum of extremal rays and a convex combination of
extremal points.
Pick an arbitrary extremal point $z^*=(p^*,y^*,\beta^*)$ from the
combination.
We claim that this is also an optimal solution to (\ref{CP}). By
Claim~\ref{claim:non-neg}, it suffices to show that it corresponds to a
market equilibrium, which is equivalent to $u_{ij}\beta^*_i=p^*_j$
whenever $y^*_{ij}>0$. For a contradiction, assume
$u_{ij}\beta^*_i<p^*_j$ and $y^*_{ij}>0$ holds for an $ij\in E$. Since
$z^*$ is included in the convex combination giving $z$, every strict inequality for
$z^*$ must also be strict for $z$; this would contradict the
optimality of $z$.
\end{proof}

Since every extremal point of a rational
polyhedron is rational with polynomially bounded size, the proof
of Theorem~\ref{thm:main} is complete. 
Next we derive the bound on the values of equilibrium prices and allocation.
For this, we assume that all $u_{ij}$'s are integers, since scaling them
by a constant does not change the equilibrium. 

\begin{lemma}
Assume all utilities are integers $\le U$ and we let $\Delta:=2^{n-1}(n+3)^{n+\frac12}U^{n}$. Then there
exists equilibrium prices $p$ that are quotients of two integers $\le
\Delta$, along with allocations $x$ that are quotients of two integers $\le \Delta^2$.
\end{lemma}

\begin{proof}
From Claim \ref{claim:extremal}, an optimal solution to (\ref{CP}) is achieved at an extremal point, say $z^*$, of the associated polyhedron. 
Let $m$ denote the number of non-zero $y_{ij}$'s at $z^*$. 
%
We claim that $m\le 2n-1$.
Indeed, consider the bipartite graph $(A,A,E')$, where $E'=\{(i,j)\ |\
y_{ij}>0\}$, $|E'|=m$. If this graph contains a cycle, then the
$y_{ij}$'s can be modified such that every binding
constraint remains binding and we get one more pair $(i,j)$ with
$y_{ij}=0$, in a contradiction with $v$ being a vertex. 

Let $Cz=b$ denote a
 subset of binding constraints for $z^*$ in the
linear system defining the feasible region of 
(\ref{CP}), after removing the columns corresponding to the $y_{ij}=0$
variables. The number of columns is $m+2n\le 4n-1$.
Note that the $2n$ equalities corresponding to the nodes are linearly
dependent, and therefore the rank of the matrix $C$ is at most $m+2n-1$.

By Cramer's rule, every $y_{ij}$, $p_j$ and $\beta_i$ is quotient of two
integers bounded by the maximum sub-determinant of $(C,b)$. 
Using Hadamard's bound, this is at most the product of the largest
$(m+2n-1)$ column norms of $(C,b)$. 
Note that $||b||\le \sqrt{n}<\sqrt{n+3}$, as the only constraints containing
nonzero constants are the $p_i\ge 1$ inequalities.
The norm  of each of $m$ the columns corresponding to the  $y_{ij}$
variables is $\sqrt{2}$ as each $y_{ij}$ is contained in two constraints with
coefficient 1. Similarly, the norm of each of the $n$ columns
corresponding to the 
$p_i$'s is at most $\sqrt{n+3}$, and the norm of each of $n$ columns
corresponding to the $\beta_i$'s is at most $\sqrt{n}U$. We need the
largest $m+2n-1$ columns and therefore may remove one of those of  norm
2.
From this, we 
can conclude that every $p_j$ and $y_{ij}$ is quotient of two integers bounded by $\Delta$.
Since the allocation $x_{ij}=y_{ij}/p_j$, we get that every $x_{ij}$ is quotient 
of two integers bounded by $\Delta^2$.
\end{proof} 

\begin{remark}
The above bound can be further strengthened to $\Delta=n!U^n$. 
\end{remark}

\section{Relation to previous work}\label{sec:general}
\subsection{Existence results}
The Arrow-Debreu market is traditionally formulated in a more general
setting. Besides the set of agents $A$, there is a set of goods $G$,
and each agent arrives to the market with an initial endowment
$w_{ig}\ge 0$ of good $g$. A market is given as ${\cal M}=(A,G,u,w)$.
Our setting corresponds to the special case
when $G=A$, and $w_{ij}=1$ if $i=j$ and 0 otherwise. We shall refer to
our special case as {\em bijective markets}.

Again, a market equilibrium consists of  prices $p:G\rightarrow
\R_{>0}$ and allocations of goods $x_{ijg}: A\times A\times
G\rightarrow \R_+$, where $x_{ijg}$ represents the amount of good $g$
sold by agent $j$ to agent $i$ such that:
\begin{itemize}
\item $\sum_{i\in A} x_{ijg}=w_{jg}$, $\forall j\in A,g\in G$, i.e., every
  good of every agent is fully sold.
\item For every $i\in A$, whenever $x_{ijg}>0$ for some $g\in G$ and
  $j\in A$, then $u_{ig}/p_g$ is the
  maximal value over $g\in G$.
\item $\sum_{j\in A,g\in G}x_{ijg}p_g=\sum_{g\in G}w_{ig}p_g$,
  $\forall i\in A$,
  that is, the money spent by agent $i$ equals his income.
\item $p_i>0$ for every $i\in A$.
\end{itemize}

The general case can be easily reduced to bijective markets (see
e.g. Jain \cite{jain}). First if a good is included in the initial
endowment of multiple agents, we give a different name for each such
occurrence. If an agent has $k$ goods in the endowment, we split the
agent into $k$ copies with the same utility function, each owning one
of the goods.

%

Consider now a market in the general form ${\cal M}=(A,G,u,w)$. 
 We say that a subset $S$ of agents is {\em
self-sufficient} whenever $u_{ig}>0$, for some $i\in S$ implies that $w_{i'g}=0, \forall
i'\in{A\setminus{S}}$. That is, agents in $S$ are not interested in the goods
owned by agents not in $S$.  We say that a market is {\em irreducible}
if there exists no self-sufficient proper subset of the agents. The following sufficient
condition was given by Gale in 1957:
\begin{theorem}[\cite{gale57}]\label{thm:gale-weak}
If the market ${\cal M}=(A,G,u,w)$ is irreducible then there exists an equilibrium.
\end{theorem}
The above condition is sufficient but not necessary.
Later, in 1976
Gale \cite{gale} gave a strengthening of the above theorem.
We say that a subset $S$ of agents is {\em super self-sufficient} if in
addition to above, $\exists i\in S$ such that $w_{ig}>0$ and
$u_{i'g}=0, \forall i'\in S$. That is, an agent in $S$ owns a good for which
no agent in $S$ is interested.

\begin{theorem}\label{thm:gale}[Existence Theorem \cite{gale}]
There exists an equilibrium in the market ${\cal M}=(A,G,u,w)$ if and
only if no subset of $A$ is super self-sufficient. 
\end{theorem}
We show that in our special case of bijective markets (i.e. $G=A$, and
$w_{ij}=1$ if $i=j$ and 0 otherwise), the existence condition in
Theorem~\ref{thm:main} is equivalent to that in
Theorem~\ref{thm:gale}.

\begin{lemma}\label{lem:self-suff}
A bijective market is irreducible if and only if the directed graph $(A,E)$ is
strongly connected. Further, (\ref{cond:suff}) holds if and only if no subset of $A$
is super self-sufficient.
\end{lemma}
\begin{proof}
The first part follows since in a bijective market a subset
$S\subseteq A$ of agents is self-sufficient
if and only if no arc enters $S$ in the directed graph $(A,E)$.
For the second part, 
assume first
that (\ref{cond:suff}) is violated for node $k$, and let $T$ denote the set of
nodes different from $k$ that can be reached on a directed path in $E$ from
$k$. Now let $S=T\cup\{k\}$.  It is easy to check that $S$ is super
self-sufficient, since $w_{kk}>0$ and $u_{ik}=0, \forall i\in S$. 

Conversely, assume there exists a super self-sufficient set $S$. According to
the condition, there exist $k \in S$, such that $w_{kk}>0$ and $u_{ik}=0,
\forall i \in S$. Clearly $k$ is a singleton component with no self-loop in
the strongly connected components of graph $(A,E)$, verifying
(\ref{cond:suff}).  
\end{proof}

\subsection{Previous convex programs}\label{sec:jain}
Let us first exhibit Cornet's convex program \cite{cornet89}. It was originally given for the 
general case of arbitrary endowments, but we present it here for bijective markets. Also, it was originally 
formulated with a max-min objective over the feasible region $\sum_{i} x_{ij}\le 1$ for all $j\in A$, $x\ge 0$; we unfold the max-min objective here in the natural way. The variable $x_{ij}$ corresponds to the amount of good $j$ purchased by agent $i$, whereas $q_i$ corresponds to the logarithm of the price of good $i$.

\begin{equation}\tag{CP-C}\label{CPC}
\begin{aligned}
\max &\ t\\
t&\le \sum_{k:ik\in E} u_{ik}x_{ik} - u_{ij}e^{q_i-q_j} \quad\forall ij\in E\\
 \sum_{j:ji\in E}x_{ji} &\le 1 \quad\forall i\in A\\
 x&\ge 0
\end{aligned}
\end{equation}

\begin{theorem}[\cite{cornet89}]\label{thm:cornet}
If (\ref{CPC}) is bounded then $t=0$, and $(t,x,q)$ is an optimal solution if and only if $(x,p)$ corresponds to a market 
equilibrium where $p_i=e^{q_i}$ for all $i\in A$. 
Further, if the market is irreducible then (\ref{CPC}) is bounded.
\end{theorem}
The proof  uses a nontrivial argument on
Lagrangian duality. Note that the existence of equilibrium follows on
under Gale's sufficient condition from 1957
(Theorem~\ref{thm:gale-weak}), as opposed to (\ref{CP}), where it
follows under the necessary and sufficient condition in Theorem~\ref{thm:gale}.

According to  Theorem~\ref{thm:cornet} and Lemma~\ref{lem:self-suff},
if the market is irreducible then $t=0$, and $\sum_{j:ji\in E}x_{ji}= 1$ must hold for every $i\in A$. By taking logarithms  we get that the following convex program has a feasible solution:
\begin{equation}\tag{CP-J}\label{CPJ}
\begin{aligned}
q_i-q_j&\le \log\left(\sum_{k:ik\in E} u_{ik}x_{ik}\right)
-\log u_{ij}\quad\forall ij\in E\\
 \sum_{j:ji\in E}x_{ji} &= 1 \quad\forall i\in A\\
x&\ge 0
\end{aligned}
\end{equation}
This is precisely the 
 convex program by Nenakov and Primak \cite{nenakov83}, and 
by Jain \cite{jain}.

\medskip

We can write the Lagrangian dual of our program (\ref{CP}), see Boyd
and Vandenberghe \cite{boyd04}. This gives
\begin{equation}\tag{CP-D}\label{CPD}
\begin{aligned}
\max &\ \sum_{i\in A}{\tau_i}\\
\delta_i-\delta_j+\tau_i \le 1 - &\sum_{k:ki\in E} w_{ki} + \log \left(\sum_{k:ik\in
  E}u_{ik} w_{ik}\right) - \log{u_{ij}} \quad \forall ij\in E \\
 \tau,w&\ge 0
\end{aligned}
\end{equation}
Note that the variables in an optimal solution correspond to optimal
Lagrangian multipliers satisfying the KKT-conditions (\ref{kkt-1})-(\ref{kkt-3}).
Theorem~\ref{thm:main} implies that strong duality holds: if
(\ref{CP}) is feasible then there exists a market equilibrium, that
easily provides a solution (\ref{CPD}).

Despite certain similarities, this formulation appears to be different
from (\ref{CPC}), namely, it has a larger feasible region.
Indeed, for every feasible solution of (\ref{CPC}), $\delta=q$, $w=x$, $\tau_i=t$
gives a feasible solution to (\ref{CPD}). Nevertheless, the converse
is not true since $\sum_{i:ij\in E}w_{ij}\le 1$
may not hold for feasible solutions of (\ref{CPD}).

We further note that following the argument of Section~\ref{sec:main-proof},
we can derive the feasibility of (\ref{CPC}). 
It follows that in an optimal solution we must have $\sum_{j:ji\in
  E}w_{ji} = 1$ and  $\tau_i = 0$ for all $i \in A$.
Using these, we can substitute $x=w$, $q=\delta$. 
This yields a
feasible solution to (\ref{CPC}).

\bibliographystyle{abbrv}
\bibliography{arrow-debreu}

\end{document}